\documentclass[12pt]{iopart}
\usepackage{graphicx}
\expandafter\let\csname equation*\endcsname\relax

\expandafter\let\csname endequation*\endcsname\relax

\usepackage{amsmath}           
\usepackage{amsfonts}          
\usepackage{epstopdf}
\usepackage{amsthm}

\usepackage{subcaption}
\usepackage{float}

\usepackage{tikz,pgfplots,bbm,bm,mathrsfs,dsfont}
\usetikzlibrary{shapes,snakes,arrows,pgfplots.groupplots,fadings,calc,decorations.markings,positioning,chains,fit,shapes,calc}

\tikzset{
	>=triangle 45,
	box/.style={draw = black, rectangle, rounded corners, inner sep=4pt,fill=white,text=black, text centered},
	arro/.style={line width=1pt,gray,->,shorten <=2pt,shorten >=1pt}
}

\usepackage{color}
\usepackage{quoting}

\usepackage[english]{babel}
\usepackage{hyperref}
\hypersetup{pdfauthor={Caracciolo DiGioacchino Malatesta Vanoni},pdftitle={Average optimal cost for the Euclidean TSP in one dimension},%
	colorlinks, linktocpage=true, pdfstartpage=3, pdfstartview=FitV%
	urlcolor=orange, linkcolor=blue, citecolor=red, 
}

\usepackage{dsfont}

\newcommand{\be}{\begin{equation}}
\newcommand{\ee}{\end{equation}}
\def\reff#1{(\protect\ref{#1})}
\newcommand{\Mod}[1]{\ (\mathrm{mod}\ #1)}

\newtheorem{lemma}{Lemma}
\newtheorem{defn}{Definition}[section]
\newtheorem{pros}{Proposition}[section]
\newtheorem{corollary}{Corollary}

\usepackage[header]{appendix}

\begin{document}
	
\title{Average optimal cost for the Euclidean TSP in one dimension}

\author{Sergio Caracciolo$^1$, Andrea Di Gioacchino$^1$,
	Enrico M. Malatesta$^1$ and Carlo Vanoni$^1$ }
\address{$^1$ Dipartimento di Fisica, University of Milan and INFN, via Celoria 16, 20133 Milan, Italy}
\ead{sergio.caracciolo@mi.infn.it, andrea.digioacchino@unimi.it, enrico.m.malatesta@gmail.com and carlo.vanoni@studenti.unimi.it}


\begin{abstract}
The traveling-salesman problem is one of the most studied combinatorial optimization problems, because of the simplicity in its statement and the difficulty in its solution. 
We study the traveling salesman problem when the positions of the cities are chosen at random in the unit interval and the cost associated to the travel between two cities is their distance elevated to an arbitrary power $p\in \mathbb{R}$.
We characterize the optimal cycle and compute the average optimal cost for every number of cities when the measure used to choose the position of the cities is flat and asymptotically for large number of cities in the other cases. We also show that the optimal cost is a self-averaging quantity, and we test our analytical results with extensive simulations.
\end{abstract}

\submitto{\JPA}


\section{Introduction}\label{I}

Since many years~\cite{Kirkpatrick1981, Kirkpatrick1983, KS, Sourlas1986, FA, BSS} we know that there is a deep relation between combinatorial optimization problems and statistical mechanics of disordered systems at zero temperature.
This relation has been fully exploited in the case of random problems defined on  graphs
with weights associated to the edges of the graph, when these one are independent, identically distributed random variables~\cite{Vannimenus1984, Orland1985, Mezard1985, Mezard1986, Mezard1986a, Mezard1987, mezard1987spin, Krauth1989, mezard2009information}.
With this version of randomness there is no effect due to the geometry of the underlying space.
In the Euclidean problems, instead, the graph is embedded in a compact subspace of $\mathbb{R}^d$ and the  positions of the vertices of the graph are independent random variables while the weights associated to the edges are, generically, a function of the distance between the corresponding vertices. Therefore these weights are now correlated because of the Euclidean geometry. These correlations prevent us to obtain typical properties of the systems under analysis by exploiting the standard replica (or cavity) methods. Only in very high dimension $d$ one can hope to deal with correlations as perturbation in the solution of the random-link case~\cite{Mezard1988}.
On the opposite side, the case $d=1$ can be simple enough to allow exact solutions, at least when the cost function is a convex function of its argument, i.e. the length of the edges.
This is the case for several random Euclidean optimization problems on the complete bipartite graph: the matching problem (or 1-factor, that is dimer covering)~\cite{Caracciolo:159, Caracciolo:160, Caracciolo:169, Sicuro,Selberg2018}, the 2-factor problem (or 2-matching that is loop covering)~\cite{Caracciolo:173,Selberg2018} and finally the traveling salesman  problem (that is covering of the lattice with a unique loop. i.e. an Hamiltonian cycle)~\cite{Caracciolo:171,Selberg2018}. Interestingly enough, in all these problems, the optimal total costs are not self-averaging, and their values are simply related. The understanding of the $d=1$ solution has been useful to allow some generalization in higher dimension, at least for some particular choices of the cost function~\cite{Caracciolo:158, Caracciolo:162, Caracciolo:163, Ambrosio2016, Caracciolo:174}.

In the case $d=1$ for the complete (monopartite) graph, only the matching problem has been analyzed~\cite{Caracciolo:169}. 
In this paper we will study the traveling salesman problem (TSP) on the complete graph embedded in the unit interval $[0,1]$, with open boundary conditions and a very general form of the cost function. We will derive the average cost for the optimal solution and look for the finite-size corrections.

From the point of view of computational complexity, the matching problem is simple, because there is an algorithm which finds the optimal solution also in the worst case in a time which scales polynomially with the number of vertices. At variance, the TSP is considered the archetypal ``hard'' problem in combinatorial optimization~\cite{lawler1985}. Nevertheless it can be stated in very simple terms: given $N$ cities and $N (N-1)/2$ values that represent the cost paid for traveling between all pairs of them, the TSP  consists in finding the tour that visits all the cities and finally comes back to the starting point with the least total cost to be paid for the journey.
The Euclidean TSP, where the costs to travel from one city to another are the Euclidean distances between them, remains NP-complete~\cite{papadimitriou1977}. 
A closed formula for the average optimal cost for the TSP for the random-link model, in the monopartite case, has been obtained only using cavity method~\cite{Krauth1989} since with replica method some technical problems arise~\cite{Mezard1986}. For this reason an analytical study of the finite-size corrections in the random TSP, where replicas are a natural tool, is still lacking. This is contrary to what happens to similar problems like the matching problem, where, in the random-link approximation these issues do not arise~\cite{Mezard1985}, and one can compute easily finite-size corrections~\cite{Mezard1987, Parisi2002}, also for quite generic disorder distribution of the link variables~\cite{Caracciolo:168}.

The paper is organized as follows. In Sec.~\ref{II} we define the combinatorial optimization problem in a generic graph. We introduce the random-link and the Euclidean stochastic version of the problem. In Sec.~\ref{III} we give a complete discussion of what are the possible optimal solutions in the different instances of the problem by changing the cost function. The evaluation of the average total cost and of its variance is presented and discussed in Sec.~\ref{IV}.  Some conclusion are given in Sec.~\ref{V}. Technical details have been collected, as far as possible, in the appendices.

\section{The problem}\label{II}
Given a generic (undirected) graph $\mathcal{G} = (\mathcal{V}, \mathcal{E})$, a {\em cycle}  of length $k$ is a sequence of edges $e_1, e_2, \dots, e_k \in \mathcal{E}$ in which two subsequent edges $e_i$ and $e_{i+1}$ share a vertex for $i=1, \dots, k$; in the $i=k$ case the edge $e_{k+1}$ must be identified with the edge $e_1$. The cycle is simple if it touches every vertex only once, except for the starting and ending vertices. From now on when we will use the word cycle to mean a simple cycle. A cycle is said to be {\em Hamiltonian} when it passes through all the vertices of the graph. 
A graph that contains an Hamiltonian cycle is called an Hamiltonian graph. Determining if a graph is Hamiltonian is an NP-complete problem~\cite{lawler1985}.
The complete graph with $N$ vertices $\mathcal{K}_N$ is Hamiltonian for $N>2$.
The bipartite complete graph with $N+M$ vertices $\mathcal{K}_{N,M}$ is Hamiltonian for $M=N>1$.

We denote by $\mathcal H$ the set of Hamiltonian cycles of the graph $\mathcal{G}$. Let us assign to each edge $e \in \mathcal{E}$ of the graph $\mathcal{G}$ a positive weight $w_e > 0$. We can associate to each Hamiltonian cycle $h\in \mathcal{H}$ a total cost
\be
E(h) :=  \sum_{e\in h} w_e \, .\label{E}
\ee
In the (weighted) Hamiltonian cycle problem we search for the Hamiltonian cycle $h\in \mathcal{H}$ that minimizes the total cost in~\reff{E} i.e.
\be
E(h^*) = \min_{ h\in \mathcal{H}} E(h)  \, , \label{h^*}
\ee
where we have denoted by $h^*\in \mathcal{H}$ the optimal Hamiltonian cycle. The search of the Hamiltonian cycle $h^*$ that minimizes the total cost is also called \emph{traveling salesman problem} (TSP), since the $N$ vertices of the graph can be seen as cities and the weight for each edge can be interpreted as the cost paid to travel the route distance between them. \\
One can introduce in different ways the weights. 
For example, consider when the graph $\mathcal{K}_N$ is embedded in $\mathbb{R}^d$, that is for each $i\in [N] =\{1,2,\dots,N\}$ we associate a point $x_i\in \mathbb{R}^d$,  and for $e=(i,j)$ with $i,j \in [N]$ we introduce a cost which is a function of their Euclidean distance 
\be
w_e=|x_i-x_j|^p \, , \qquad p\in \mathbb{R} .\label{cost}
\ee
When $p=1$, we obtain the usual \emph{Euclidean} TSP.\\
In statistical physics, one is interested in introducing randomness in the problem to study the typical properties of the optimal solution, as for example the average optimal cost
\be
\overline{E}  := \overline{E(h^*)}\,.
\ee
We have denoted by a bar the average over the disorder. 
The simplest way to introduce it is to consider the weights $w_e$ independent and identically distributed random variables. In this case the problem is called {\em random} TSP and has been previously studied using the replica and the cavity method ~\cite{Vannimenus1984,Orland1985,Sourlas1986,Mezard1986,Mezard1986a,Krauth1989,Ravanbakhsh2014}. In particular, the results for the average optimal cost obtained were also proved rigorously~\cite{Wastlund}. In the random Euclidean TSP case~\cite{BHH,Steele,Karp,Percus1996,Cerf1997}, instead, one extracts randomly the positions of the points. As a consequence, this is a more challenging problem due to the presence of correlations between weights. The weights $w_e$ can be chosen to be, in general, a function of the Euclidean distance between two points and, in this work, we will use the function given in equation (\ref{cost}).

Another problem related to the TSP that we have introduced in Sec. \ref{I} is the so called \emph{2-factor} problem. It can be seen as a relaxation of the TSP since it does not have the single loop condition. In other words, the 2-factor problem corresponds to find the minimum loop-covering of the graph. In the following we will denote with $\mathcal M_2$ the set of all possible 2-factors of the graph $\mathcal{G}$.

\section{Optimal Cycles on the complete graph}~\label{III}

We shall consider the complete graph $\mathcal{K}_{N}$ with $N$ vertices, that is with vertex set $V=[N]:=\{1,\dots, N\}$. 
This graph has $\frac{(N -1)!}{2}$ Hamiltonian cycles. 
Indeed, each permutation $\pi$ in the symmetric group of $ N $ elements, $\pi \in \mathcal{S}_{N}$, defines an Hamiltonian cycle on $\mathcal{K}_{N}$. 

The sequence of points $(\pi(1), \pi(2),\dots,\pi(N),\pi(1))$ defines a closed walk with starting point $\pi(1)$, but the same walk is achieved by choosing any other vertex as starting point and also by following the walk in the opposite order, that is, $(\pi(1), \pi(N),\dots,\pi(2),\pi(1))$. As the cardinality of $\mathcal{S}_{N}$ is $N!$ we get that the number of Hamiltonian cycles in $\mathcal{K}_{N}$ is $\frac{N!}{2  N}$.

In this section, we characterize the optimal Hamiltonian cycles for different values of the parameter $p$ used in the cost function. Notice that $p=0$ and $p=1$ are degenerate cases, in which the optimal tour can be found easily.

\subsection{The $p>1$ case}

We start by proving which is the optimal cycle when $p>1$, for every realization of the disorder. To do so, we will summarize here, first, some results, already obtained in ~\cite{Caracciolo:171}, for the TSP defined on a bipartite complete graph $\mathcal{K}_{N,N}$ embedded on the interval $\Omega =[0,1]\subset \mathbb{R}$, which will be useful in the following.

Let us call $\mathcal{R} := \{r_i\}_{i=1,\dots,N}\subset\Omega$ and $\mathcal{B}:=\{ b_j\}_{j=1,\dots,N}\subset\Omega$ two sets of points, both of cardinality $N$, in the interval $\Omega$. All points are supposed to be realizations of random variables which are identically and independently distributed over $\Omega$, according to the flat distribution. 

Remember that the bipartite complete graph $\mathcal{K}_{N,N}$ has $\frac{N! \, (N-1)!}{2}$ Hamiltonian cycles.
Indeed, given two permutations $\sigma, \pi\in \mathcal{S}_n$ the sequence of points
\begin{equation}
	\begin{split}
		& h[(\sigma, \pi)] := ( r_{\sigma(1)} , b_{\pi(1)}, r_{\sigma(2)}, b_{\pi(2)}, \dots, r_{\sigma(N)} , b_{\pi(N)}, r_{\sigma(1)})
	\end{split}
\end{equation}
defines a closed walk starting from $\sigma(1)$ and therefore if we fix $\sigma(1) = 1$ this sequence or, equivalently, the sequence obtained by reversing the order defines an Hamiltonian cycle.

In ~\cite{Caracciolo:171} we proved that, if both red and blue points are ordered, i.e. $r_1 \le \dots \le r_N$ and $b_1 \le \dots \le b_N$, the optimal cycle is defined by $h^* = \tilde{h} := h[(\tilde{\sigma}, \tilde{\pi})]$ with 
\be
\tilde{\sigma}(i) := 
\begin{cases}
	2i-1 & i \leq  (N+1)/2 \\
	2N -2i +2 & i > (N+1)/2 \label{sigmatilde}
\end{cases}
\ee
and
\begin{equation}
	\tilde{\pi}(i) = \tilde{\sigma}(N+1-i)  \label{pitilde}
\end{equation}
which is exactly the reverse of $\tilde{\sigma}$. Therefore at fixed ordered set of red and blue points the optimal cost is exactly
\be
\begin{split}
	E_N(\tilde{h}) & = |r_1-b_1|^p + |r_N-b_N|^p  +\sum_{i=1}^{N-1}\left[ |b_{i+1} - r_i|^p  + |r_{i+1} - b_i|^p \right] \,,
\end{split}
\label{EN}
\ee
for every realization of the disorder.\\
Let us suppose, now, to have $N$ points $\mathcal{R} = \{r_i\}_{i=1,\dots,N}$ in the interval $[0,1]$. As usual we will assume that the points are ordered, i.e. $r_1 \le \dots \le r_N$. Let us define the following Hamiltonian cycle
\be
h^* = h[\tilde{\sigma}] = (r_{\tilde{\sigma}(1)}, r_{\tilde{\sigma}(2)}, \dots, r_{\tilde{\sigma}(N)}, r_{\tilde{\sigma}(1)})
\ee
with $\tilde{\sigma}$ defined as in~\reff{sigmatilde}. In~\ref{app:proofs} we prove that
\begin{pros}\label{pro1}
	The Hamiltonian cycle which provides the optimal cost is $h^*$.
\end{pros}
The main ideas behind the proof is that we can introduce a complete bipartite graph in such a way that a solution of the bipartite matching problem on it is a solution of our original problem, with the same cost. Therefore, using the results known for the bipartite problem, we can prove Proposition~\ref{pro1}. \\

A graphical representation of the optimal cycle for $p>1$ and $N=6$ is given in Fig. \ref{fig1}, left panel.
\begin{figure*}[t]
	\begin{subfigure}{0.5\linewidth}
		\centering
		\includegraphics[width=0.9\columnwidth]{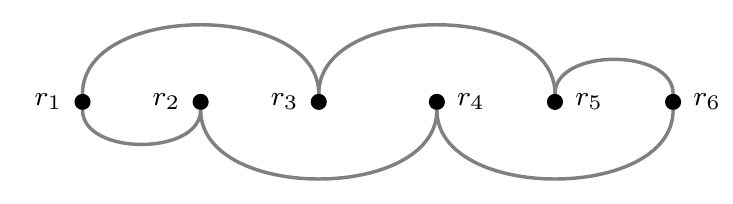}
	\end{subfigure} \hfill
	\begin{subfigure}{0.5\linewidth}
		\centering
		\includegraphics[width=0.9\columnwidth]{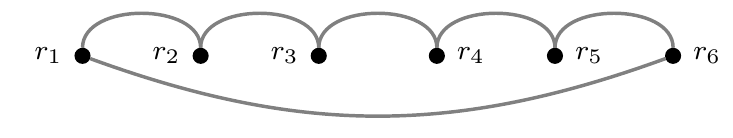}
	\end{subfigure}
	\caption{Optimal solutions for $N=6$, for the $p>1$ (left) and $0<p<1$ (right) cases. Notice how, in the $0<p<1$ case, when all the arcs are drawn in the upper half-plane above the points than there is no crossing between the arcs.}\label{fig1}
\end{figure*}
 
\subsection{The $0<p<1$ case}
We now prove that, given an ordered sequence $\mathcal{R} = \{r_i\}_{i=1,\dots,N}$ of $N$ points in the interval $[0,1]$, with $r_1 \le \dots \le r_N$, if $0 < p <1$ and if
\be
h^* = h[{\mathds{1}}] = (r_{{\mathds{1}}(1)}, r_{{\mathds{1}}(2)}, \dots, r_{{\mathds{1}}(N)}, r_{{\mathds{1}}(1)})
\ee
where $\mathds{1}$ is the identity permutation, i.e.:
\be
\mathds{1}(j) = j 
\ee
then
\begin{pros}
	\label{Proposition::0<p<1}
	The Hamiltonian cycle which provides the optimal cost is $h^*$.
\end{pros}

The idea behind this result is that we can define a crossing in the cycle as follows: let $\{r_i\}_{i=1,\dots,N}$ be the set of points, labeled in ordered fashion; consider two links $(r_i,r_j)$ and $(r_k,r_\ell)$ with $i<j$ and $k<\ell$; a crossing between them occurs if $i<k<j<\ell$ or $k<i<\ell<j$.
This corresponds graphically to a crossing of lines if we draw all the links as, for example, semicircles in the upper half-plane. In the following, however, we will not use semicircles in our figures to improve clarity (we still draw them in such a way that we do not introduce extra-crossings between links other than those defined above). An example of crossing is in the following figure

\begin{figure}[h]
	\centering
	\includegraphics[width=0.4\columnwidth]{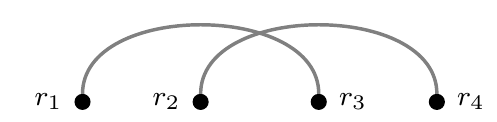}
\end{figure}
\noindent where we have not drawn the arcs which close the cycle to emphasize the crossing. Now, as shown in~\cite{Caracciolo:159}, if we are able to swap two crossing arcs with two non-crossing ones, the difference between the cost of the original cycle and the new one simply consists in the difference between a crossing matching and a non-crossing one,  that is positive when $0<p<1$. Therefore the proof of Proposition \ref{Proposition::0<p<1}, which is given in~\ref{app:proofs}, consists in showing how to remove a crossing (without breaking the cycle into multiple ones) and in proving that $h^*$ is the only Hamiltonian cycle without crossings (see Fig. \ref{fig1}, right panel).

\subsection{The $p<0$ case}\label{sec:p<0}
Here we study the properties of the solution for $p<0$. Our analysis is based, again, on the properties of the $p<0$ optimal matching solution. In~\cite{Caracciolo:169} it is shown that the optimal matching solution maximizes the total number of crossings, since the cost difference of a non-crossing and a crossing matching is always positive for $p<0$. This means that the optimal matching solution of $2N$ points on an interval is given by connecting the $i$-th point to the $(i+N)$-th one with $i = 1, \dots, N$; in this way every edge crosses the remaining $N-1$. 
Similarly to the $0<p<1$ case, suppose now to have a generic oriented Hamiltonian cycle and draw the connections between the vertices in the upper half plain (as before, eliminating all the crossings which depend on the way we draw the arcs). Suppose it is possible to identify a matching that is non-crossing, then the possible situations are the following two (we draw only the points and arcs involved in the non-crossing matching):
\begin{figure}[h]
	\centering
	\includegraphics[width=0.4\columnwidth]{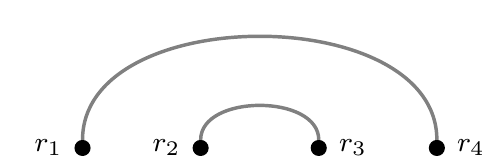} 
\end{figure}
\begin{figure}[h]
	\centering
	\includegraphics[width=0.4\columnwidth]{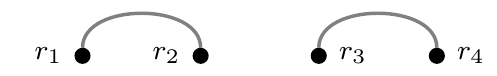}
\end{figure}

\noindent In~\ref{app:proofs}, we discuss the move that allows to replace non-crossing matchings by crossing ones, in such a way that the cycle that contains the matching remains an Hamiltonian cycle.
This move is such that the cost of the new configuration is lower than the cost of the old one, since the cost gain is the difference between the costs of a non-crossing and a crossing matching, which is always positive for $p<0$.

In this manner the proof for $p<0$ goes on the same line of $0<p<1$, but instead of finding the cycle with no crossings, now we look for the one or ones that maximize them. However, as we will see in the following, one must distinguish between the $N$ odd and even case. In fact, in the $N$ odd case, only one cycle maximizes the total number of crossings, i.e. we have only one possible solution. In the $N$ even case, on the contrary, the number of Hamiltonian cycles that maximize the total number of crossings are $\frac{N}{2}$. 

\paragraph{N odd case}

Given an ordered sequence $\mathcal{R} = \{r_i\}_{i=1,\dots,N}$ of $N$ points, with $N$ odd, in the interval $[0,1]$, with $r_1 \le \dots \le r_N$, consider the permutation $\sigma$ defined as:
\be\label{sigma}
\sigma(i) = 
\begin{cases}
	1 & \hbox{for } i = 1 \\
	\frac{N-i+3}{2} & \hbox{for } $even $ i>$1$ \\
	\frac{2N-i+3}{2} & \hbox{for } $odd $ i>$1$
\end{cases}
\ee
This permutation defines the following Hamiltonian cycle:
\be
h^*:=h[\sigma]=(r_{\sigma(1)}, r_{\sigma(2)},\dots,r_{\sigma(N)}) .\label{opt1}
\ee
\begin{pros}\label{prop<01}
	The Hamiltonian cycle which provides the optimal cost is $h^*$.
\end{pros}
The proof consist in showing that the only Hamiltonian cycle with the maximum number of crossings is $h^*$. As we discuss in~\ref{app:proofs}, the maximum possible number of crossings an edge can have is $N-3$. The Hamiltonian cycle under exam has $N(N-3)/2$ crossings, i.e. every edge in $h^*$ has the maximum possible number of crossings. Indeed, the vertex $a$ is connected with the vertices $a+\frac{N-1}{2} \pmod N$ and $a+\frac{N+1}{2}  \pmod N$. The edge $(a,a+\frac{N-1}{2} \pmod N)$ has $2\frac{N-3}{2}=N-3$ crossings due to the $\frac{N-3}{2}$ vertices $a+1\pmod N, a+2\pmod N,\dots, a+\frac{N-1}{2}-1\pmod N$ that contribute with 2 edges each. This holds also for the edge $(a,a+\frac{N+1}{2} \pmod N)$ and for each $a\in[N]$. As shown in~\ref{app:proofs} there is only one cycle with this number of crossings. 

Now, notice that an Hamiltonian cycle is a particular loop covering. However, if we search for a loop covering in the $p<0$ case, we need again to find the one which maximizes the number of crossings. Since the procedure of swapping two non-crossing with two crossing arcs can be applied to each loop covering but $h^*$ it follows that:
\begin{corollary}
	$h^*$ provides also the optimal 2-factor problem solution.
\end{corollary}
An example of an Hamiltonian cycle discussed here is given in Fig. \ref{fig2}.
\begin{figure}
	\centering
	\includegraphics[width=0.5\columnwidth]{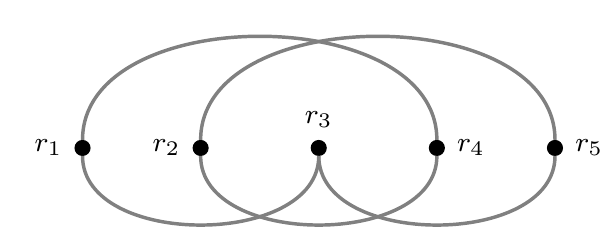}
	\caption{This is the optimal TSP and 2-factor problem solution for $N=5$, in the $p<0$ case. Notice that there are no couples of edges which do not cross and which can be changed in a crossing couple.}\label{fig2}
\end{figure}

\paragraph{N even case}

\begin{figure*}[ht]
	\begin{subfigure}[b]{0.5\linewidth}
		\centering
		\includegraphics[width=0.7\columnwidth]{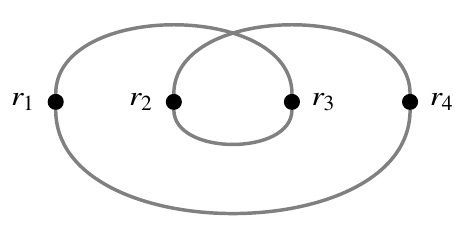}
	\end{subfigure}
	\begin{subfigure}[b]{0.5\linewidth}
		\centering
		\includegraphics[width=0.7\columnwidth]{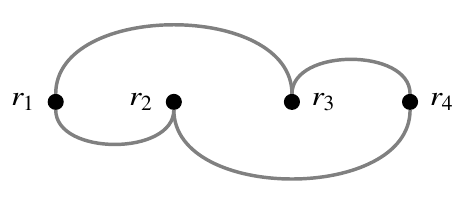}
	\end{subfigure}
	\caption{The two possible optimal Hamiltonian cycles for $p<0$, $N=4$. For each specific instance one of them has a lower cost than the other, but differently from all the other cases ($p>0$ or $N$ odd) the optimal cycle is not the same for each disorder instance. }\label{fig3}
\end{figure*}	

In this situation, differently from the above case, the solution is not the same irrespectively of the disorder instance. More specifically, there is a set of possible solutions, and at a given instance the optimal is the one among that set with a lower cost. We will show how these solutions can be found and how they are related. In this section we will use the results obtained in~\ref{B} regarding the Monopartite Euclidean 2-factor for $p<0$.\\
Given the usual sequence of points   $\mathcal{R} = \{r_i\}_{i=1,\dots,N}$ of $N$ points, with $N$ even, in the interval $[0,1]$, with $r_1 \le \dots \le r_N$, if $p <0$, consider the permutation $\sigma$ such that:
\be
\sigma(i)=\begin{cases}
	1  &\mbox{for } i =1\\
	\frac{N}{2}-i+3 &\mbox{for even } i \leq \frac{N}{2}+1  \\
	N-i+3 &\mbox{for odd } i \leq \frac{N}{2}+1\\
	i-\frac{N}{2} &\mbox{for even } i>\frac{N}{2}+1\\
	i &\mbox{for odd } i>\frac{N}{2}+1\\
\end{cases}
\ee

%
Given $\tau \in \mathcal{S}_N$ defined by $\tau(i)=i+1$ for $i \in [N-1]$ and $\tau(N)=1$, we call $\Sigma$ the set of permutations  $\sigma_k, k=1,...,N$ defined as:
\be 
\sigma_k(i)=\tau^k(\sigma(i))
\ee 
where $\tau^k=\tau \circ \tau^{k-1}$.
Thus we have the following result:
\begin{pros} \label{TSP_p<0}
The optimal Hamiltonian cycle is one of the cycles defined as
\be
h_k^*:=h[\sigma_k]=(r_{\sigma_k(1)}, r_{\sigma_k(2)},\dots,r_{\sigma_k(N)}) .\label{opt2}
\ee
\end{pros}
An example with $N=4$ points is shown in Fig. \ref{fig3}.
In~\ref{B} the optimal solution for the Euclidean 2-factor in obtained. In particular, we show how the solution is composed of a loop-covering of the graph. The idea for the proof of the TSP is to show how to join the loops in the optimal way in order to obtain the optimal TSP. The complete proof of the Proposition \ref{TSP_p<0} is given in~\ref{B}.

Before ending this section, we want to extend to more general cost functions some of our results.
As it is shown in~\cite{Caracciolo:169}, the cost function in~\reff{cost}, when $p<0$, has the same matching properties of a class of functions that were called $C$-functions:
\begin{defn}
	We will say that a function $f: [0,1]\rightarrow \mathbb{R}$ is a $C$-function if, given $0<z_1<z_2<1$, for any $\eta \in (0, 1-z_2),  \mu \in (z_2,1)$
	\be
	f(z_2)-f(z_1)\leq f(\eta+z_2)-f(\eta+z_1) \label{17}
	\ee
	\be
	f(z_2)-f(z_1)\leq f(\mu-z_2)-f(\mu-z_1) \, . \label{18}
	\ee
\end{defn}   

Eq.~\reff{17} implies that $\Psi_\eta(z) := f (\eta + z)-f (z)$ is an increasing function in the interval $(0,1-\eta)$ for any value of $\eta \in (0,1)$. Moreover, if $f$ is continuous, Eq.~\reff{17} is equivalent to convexity.
Eq.~\reff{18} implies that $\Phi_\eta(z) := f (\eta- z) - f (z)$ is increasing in the interval $(0,\eta)$ for any value of of $\eta \in (0,1)$. If $f$ is differentiable, this fact can be written as
\be
f'(\eta- z)+f'(z)\leq 0 , \quad z\in(0,\eta),  \quad \eta\in (0,1) , 
\ee
which for $\eta\to 1$ becomes
\be
f'(1-z) + f'(z) \leq 0 \, . 
\ee 
This implies, for example, that the convex function
\be
                            f_\alpha(x) :=(x-\alpha)^2, \quad \alpha \in \mathbb{R} \, , 
\ee
is a $C$-function on the unit interval for $\alpha \geq 1/2$ only.

In case the cost is a $C$-function, the expression for the weight in~\reff{E} becomes:
\be
w_e=f(|x_i-x_j|)
\ee
where $f$ is a $C$-function.

Since both in the $N$ odd case and in the $N$ even case the results obtained only depend on the properties of the matching for $p<0$, and since the weight $w_e$ in ~\reff{cost} is a $C$-function for $p<0$, our results for this case are valid in general for a $C$-function cost. It holds, then:
\begin{pros}
	Given a cost function of the form:
	\be
	E(h)=\sum_{e\in h} w_e
	\ee
	where
	\be
		w_e=f(|x_i-x_j|)
	\ee
	with $f$ a $C$-function, the optimal cycle is given by ~\reff{opt1} when N is odd and by~\reff{opt2} when N is even.
\end{pros}

\section{Analytical and numerical results}~\label{IV}
In this section we will use the insight on the solution just obtained to compute typical properties of the optimal cost for various values of $p$.

\subsection{Evaluation of the average costs}
\begin{figure}
	\centering
	\includegraphics[width=0.8\columnwidth]{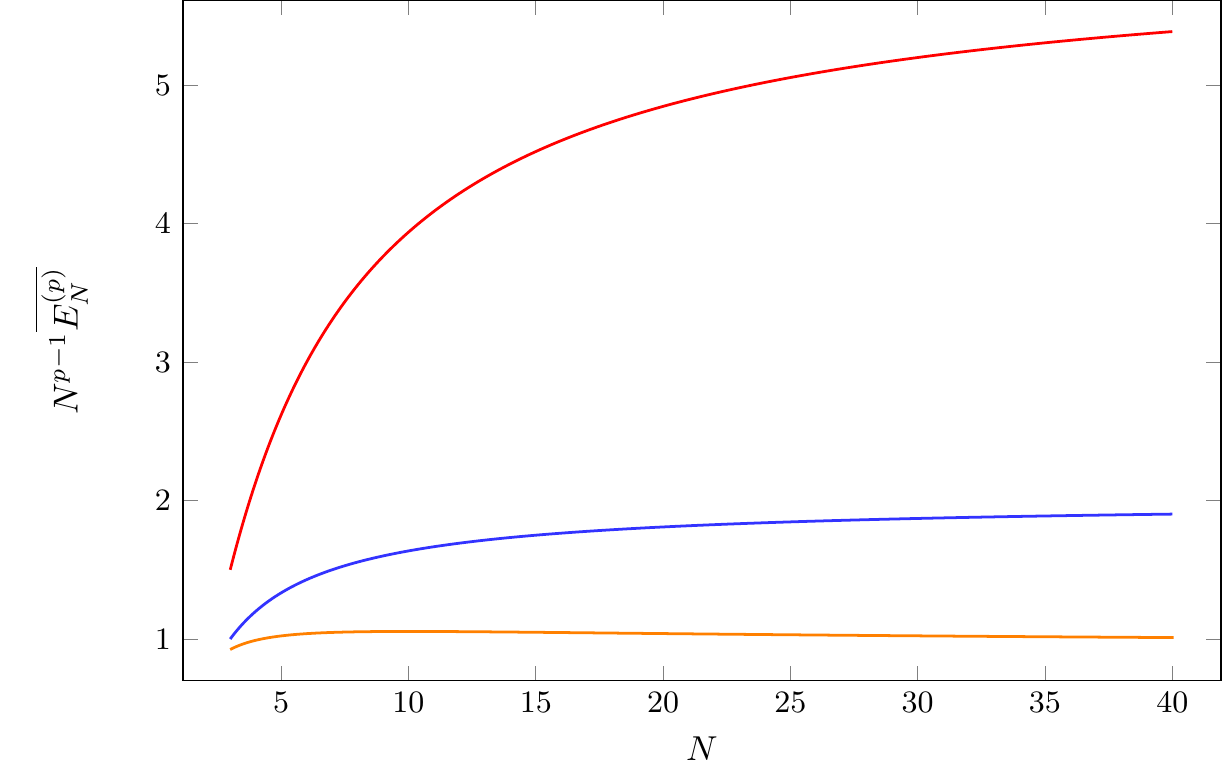} 
	\caption{Rescaled average optimal cost for $p=2$, 1, 0.5 (from top to bottom).} \label{fig::TSP1}
\end{figure}
Given $N$ random points chosen with flat distribution in the interval $[0,1]$ and we order them in increasing position, the probability for finding the $l$-th point in $[x,x+dx]$ is given by
\begin{equation}
p_l(x) dx = \frac{\Gamma(N+1)}{\Gamma(l) \, \Gamma(N-l+1)} \, x^{l-1} (1-x)^{N-l} dx
\end{equation}
while the probability for finding the $l$-th point in $[x,x+dx]$ and the $s$-th point in $[y,y+dy]$ is given, for $s>l$  by
\begin{equation}
\begin{split}
p_{l,s}(x, y) dx dy  & = \frac{\Gamma(N+1)}{\Gamma(l) \, \Gamma(s-l) \, \Gamma(N-s+1)} \, x^{l-1} \\
& \times (y-x)^{s-l-1} (1-y)^{N-s}\, \theta(y-x) dx dy \,,
\end{split}
\end{equation}
see for example~\cite[App.~A]{Caracciolo:160}. It follows that
\be
\int \,dx\, dy\, (y-x)^\alpha\, p_{l, \, l+k}(x, y) = \frac{\Gamma(N+1)\, \Gamma(k+\alpha)}{\Gamma(N+\alpha+1)\, \Gamma(k)   } 
\label{k}
\ee
independently from $l$, and, therefore, in the case $p>1$ we obtain
\be
\overline{E_N[h^*]} = \left[ (N-2) (p+1) +2 \right] \, \frac{\Gamma(N+1)\, \Gamma(p+1)}{\Gamma(N+p+1)   } 
\label{AOC}
\ee
and in particular for $p=2$ 
\be
\overline{E_N[h^*]} = \frac{ 2\, (3 N - 4)}{(N+1) (N+2)}\, ,
\ee
and for $p=1$ we get
\be
\overline{E_N[h^*]} = \frac{ 2\, (N - 1)}{N+1}\,.
\label{p1}
\ee
In the same way one can evaluate the average optimal cost when $0<p<1$, obtaining
\be
\begin{split}
\overline{E_N[h^*]} = \frac{\Gamma(N+1)}{\Gamma(N+p+1)} \left[(N-1) \, \Gamma(p+1) + \frac{\Gamma(N + p - 1)}{\Gamma(N-1)} \right]
\end{split}
\ee
which coincides at $p=1$ with~\reff{p1} and, at $p=0$, provides  $\overline{E_N[h^*]}  = N$. For large $N$, we get
\be
\lim_{N\to \infty} N^{p-1} \overline{E_N[h^*]}  = 
\begin{cases}
\Gamma(p+2) & \hbox{for \,}  p\ge 1 \\
\Gamma(p+1) & \hbox{for \,}  0< p<1\, .
\end{cases}
\ee
The asymptotic cost for large $N$ and $p>1$ is $2(p+1)$ times the average optimal cost of the matching problem on the complete graph $\mathcal{K}_N$ which has been computed in~\cite{Caracciolo:169}; notice that in~\cite{Caracciolo:169} the cost is normalized with $N$ and the number of points is $2N$, differently from what we do here. This factor $2(p+1)$ is another difference with respect to the bipartite case, where the cost of the TSP is twice the cost of the assignment problem~\cite{Caracciolo:171}, independently of $p$.
We report in~\ref{app::C} the computation of the average costs when the points are extracted in the interval $\left[0,1\right]$ using a general probability distribution.

For $p<0$ and $N$ odd we have only one possible solution, so that the average optimal cost is
\begin{equation}
\begin{split}
\overline{E_N[h^*]} = \frac{\Gamma(N+1)}{2\Gamma(N+p+1)} \left[(N-1) \frac{\Gamma\left(\frac{N+1}{2}+p\right)}{\Gamma \left( \frac{N+1}{2} \right)} + (N+1) \frac{\Gamma\left(\frac{N-1}{2}+p\right)}{\Gamma \left( \frac{N-1}{2} \right)}\right] \,.
\end{split}
\end{equation}
For large $N$ it behaves as
\begin{equation}
\lim_{N\to \infty}  \frac{\overline{E_N[h^*]}}{N} = \frac{1}{2^p} \,,
\end{equation}
which coincides with the scaling derived before for $p=0$. Note that for large $N$ the average optimal cost of the TSP problem is two times the one of the corresponding matching problem for $p<0$~\cite{Caracciolo:169}. 

\begin{figure}
	\centering
	\includegraphics[width=0.8\columnwidth]{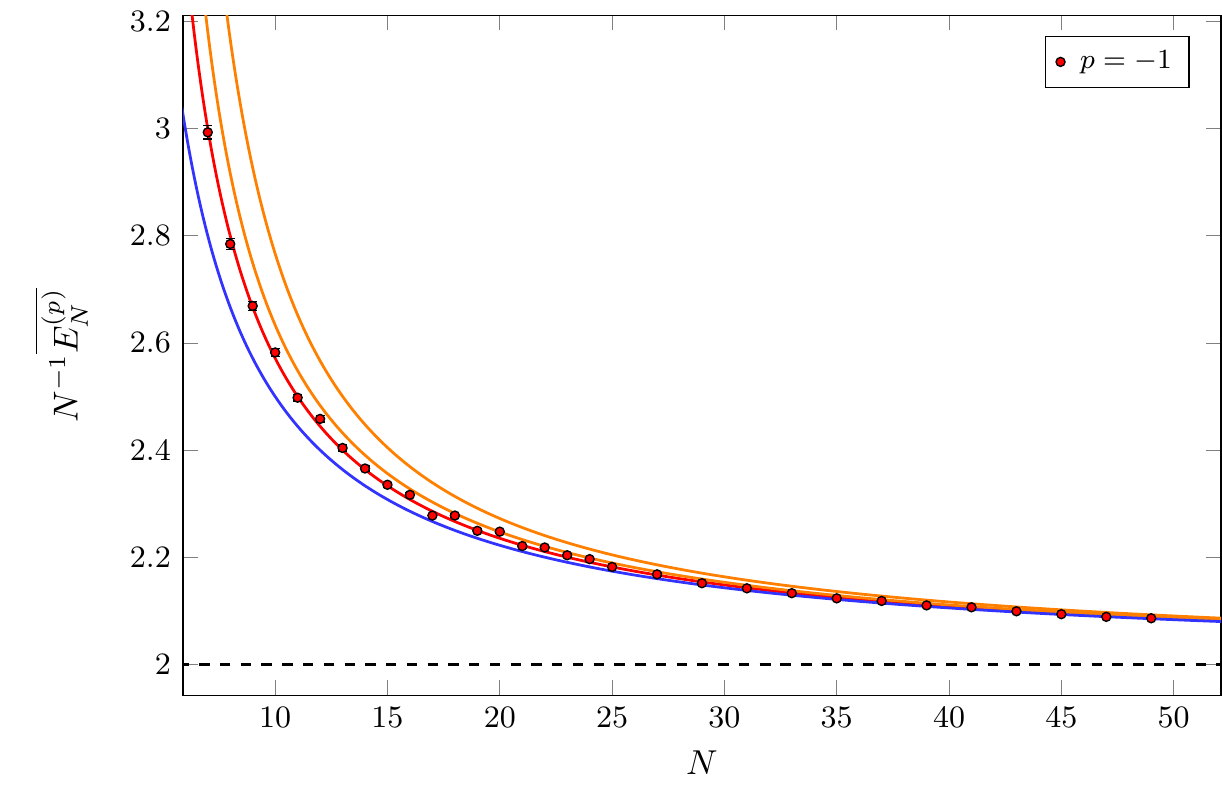} 
	\caption{Rescaled average optimal cost in the $p=-1$ case. The red points and line are respectively the result of a numerical simulation and the theoretical prediction in the odd $N$ case. The blue line is the 2 times the theoretic value of the optimal matching. The orange lines (from top to bottom) are the average costs $\overline{E_N[h_1]}$ and $\overline{E_N[h_2]}$ defined in equation (\ref{h_1}) and (\ref{h_2}) respectively. The dashed black line is the large $N$ limit of all the curves.} \label{fig::TSP2}
\end{figure}

For $N$ even, instead, there are $N/2$ possible solutions. One can see $N/2-1$ of these share the same average energy, since they have the same number of links with the same $k$ of equation (\ref{k}). These solutions have 2 links with $k=N/2$, $N/2$ links with $k=N/2+1$ and $N/2-2$ links with $k=N/2+1$. We denote this set of configurations with $h_1$ (although they are many different configurations, we use only the label $h_1$ to stress that all of them share the same average optimal cost) and its average cost is
\begin{equation}
\begin{split}
\overline{E_N[h_1]} =& \frac{\Gamma(N+1)}{\Gamma(N+p+1)} \left[ \frac{N}{2} \frac{\Gamma\left( \frac{N}{2} + p-1 \right)}{\Gamma \left( \frac{N}{2}-1 \right)}  \right.\\
& \left. + \left( \frac{N}{2}-2 \right) \frac{\Gamma \left( \frac{N}{2} + p + 1 \right)}{\Gamma \left( \frac{N}{2} +1 \right)} + 2 \frac{\Gamma\left( \frac{N}{2}+p \right)}{\Gamma \left( \frac{N}{2} \right)}\right].
\end{split}
\label{h_1}
\end{equation}
The other possible solution, that we denote with $h_2$ has 2 links with $k=N/2-1$, $N/2$ links with $k=N/2+1$ and $N/2-1$ links with $k=N/2+1$ and its average cost is
\begin{equation}
\begin{split}
\overline{E_N[h_2]} = & \frac{\Gamma(N+1)}{\Gamma(N+p+1)} \left[ \left(\frac{N}{2}-1\right) \frac{\Gamma\left( \frac{N}{2} + p-1 \right)}{\Gamma \left( \frac{N}{2}-1 \right)} \right. \\
& \left. + \left( \frac{N}{2}-1 \right) \frac{\Gamma \left( \frac{N}{2} + p + 1 \right)}{\Gamma \left( \frac{N}{2} +1 \right)} + 2 \frac{\Gamma\left( \frac{N}{2}+p \right)}{\Gamma \left( \frac{N}{2} \right)}\right].
\end{split}
\label{h_2}
\end{equation}
In Fig.~\ref{fig::TSP1} we plot the analytical results for $p=0.5$, $1$, $2$ and in Fig.~\ref{fig::TSP2} we compare analytical and numerical results for $p=-1$. In particular, since $\overline{E_N[h_1]} > \overline{E_N[h_2]}$, $\overline{E_N[h_2]}$ provides our best upper bound for the average optimal cost of the $p=-1$, $N$ even case. The numerical results have been obtained by solving $10^4$ TSP instances 
using its linear programming representation. In particular, for a given instance of the disorder, we first solve the corresponding 2-factor problem; then if the optimal 2-factor solution contains more than one loop, we add constraints in order to eliminate them; this is repeated until a solution with only one loop is achieved.

\subsection{Self-averaging property for $p>1$}
An interesting question is whether the average optimal cost is a self-averaging quantity. Previous investigation regarding the matching problem~\cite{Steele1997,Yukich2006} showed that indeed the average optimal cost is self-averaging in every dimensions when the graph is the complete one, whereas it is not in the bipartite case in dimension 1. In addition, in~\cite{Caracciolo:171} it was shown that also the average optimal cost of the one-dimensional random Euclidean TSP on the complete bipartite graph is not self-averaging. This difference between the bipartite and the complete case can be understood in terms of local fluctuations of blue and red point density, that make arguments as those given in~\cite{Houdayer1998} fail.
This is the case, at least in the one dimensional case, also for the random Euclidean TSP. Again we collect in~\ref{app::2Moment} all the technical details concerning the evaluation of the second moment of the optimal cost distribution $\overline{E_N^2}$. Here we only state the main results. $\overline{E_N^2}$ has been computed for all number of points $N$ and, for simplicity, in the case $p>1$ and it is given in equation~\eqref{SecondMoment}.
In the large $N$ limit it goes like
\be
\lim_{N\to \infty} N^{2(p-1)} \overline{E_N^2[h^*]} = \Gamma^2(p+2)
\ee
i.e. tends to the square of the rescaled average optimal cost. This proves that the cost is a self-averaging quantity. Using \eqref{SecondMoment} together with equation \eqref{AOC} one gets the variance of the optimal cost. In particular for $p=2$ we get
\be
\sigma_{E_N}^2=\frac{4 (N (5 N (N+13)+66)-288)}{(N+1)^2 (N+2)^2 (N+3) (N+4)} \,,
\label{Variance}
\ee
which goes to zero as $\sigma_{E_N}^2 \simeq 20/N^3$. 



\section{Conclusions}\label{V}
The TSP is one of the most important and studied problem in combinatorial optimization. In the present paper we have analyzed a specific version of the TSP from the point of view of statistical mechanics, which is focused on typical properties of the solution. We have considered the problem in one dimension, with a cost function which involves Euclidean distances between points to the power $p$. We have characterized the optimal cycles for every value of $p$. For $p>1$ the optimal cycles follows directly from the bipartite case~\cite{Caracciolo:171}; for $0<p<1$ the solution is the one that minimizes the number of crossings. For $p<0$, instead, the optimal cycles are the ones that maximize the number of crossings; it turns out that when the number of points $N$ to be visited is odd there is only one cycle maximizing the number of crossings; as a consequence the same cycle is the solution of the 2-factor problem. However, when $N$ is even, there are $N/2$ cycles (connected by a cyclic permutation) that have the maximum number of intersections, so that the solution depends on the particular instance of the disorder, i.e. on the position of the points. In the last part of the work we have computed the average optimal cost and we have compared theoretical prediction with numerical results. Finally we have proved that, as happens for the matching problem, the average optimal cost is a self-averaging quantity, by explicitly computing its variance.
With this work we can add the TSP to those combinatorial optimization problems in a low dimension number for which we can study typical properties, despite the absence of a general method such as the replica method, which is spoiled by the Euclidean correlations.

\appendix

\section{Proofs}\label{app:proofs}
In this appendix we prove various propositions stated in the main text.
\subsection{Proof of Proposition \ref{pro1}} \label{App::p>1}
Consider a $\sigma \in \mathcal{S}_N$ with $\sigma(1) =1$. As we said before, taking $\sigma(1) =1$ correspond to the irrelevant choice of the starting point of the cycle.
Let us introduce now a new set of ordered points $\mathcal{B}:=\{ b_j\}_{j=1,\dots,N}\subset [0,1]$
such that
\be
b_i = 
\begin{cases}
	r_1 & \hbox{for \, } i =1\\
	r_{i-1} & \hbox{otherwise } 
\end{cases}
\ee
and consider the Hamiltonian cycle on the complete bipartite graph with vertex sets $\mathcal{R}$ and $\mathcal{B}$
\be
\begin{split}
	& h[(\sigma, \pi_\sigma)] := (r_1, b_{\pi_\sigma(1)}, r_{\sigma(2)}, b_{\pi_\sigma(2)},\dots, r_{\sigma(N)}, b_{\pi_\sigma(N)}, r_{\sigma(1)})
\end{split}
\ee
so that 
\be\label{pi_s}
\pi_\sigma(i) = 
\begin{cases}
	2 & \hbox{for \, } i =1\\
	\sigma(i)+1 & \hbox{for \, } i < k\\
	\sigma(i+1)+1 & \hbox{for \, } i \geq k\\
	1 & \hbox{for \, } i = N\\
\end{cases}
\ee
where $k$ is such that $\sigma(k) = N$.
We have therefore
\be\label{b}
\begin{split}
	& (b_{\pi_\sigma(1)}, b_{\pi_\sigma(2)}, \dots, b_{\pi_\sigma(k-1)}, b_{\pi_\sigma(k)}, \dots, b_{\pi_\sigma(N-1)}, b_{\pi_\sigma(N)}) \\
	& = (r_1, r_{\sigma(2)}, \dots, r_{\sigma(k-1)},r_{\sigma(k+1)} ,\dots , r_{\sigma(N)},r_1).
\end{split}
\ee
In other words we are introducing a set of blue points such that we can find a bipartite Hamiltonian tour which only use link available in our monopartite problem and has the same cost of $\sigma$. Therefore, by construction (using \eqref{b}):
\be
\begin{split}
	E_N(h[\sigma]) & = E_N(h[(\sigma, \pi_\sigma)] ) \geq E_N(h[(\tilde{\sigma}, \tilde{\pi})]) \\
	& = E_N(h[(\tilde{\sigma}, \pi_{\tilde{\sigma}})]) = E_N(h[\tilde{\sigma}]),
\end{split}
\ee
where the fact that $\tilde{\pi}=\pi_{\tilde{\sigma}}$ can be checked using \eqref{sigmatilde} and \eqref{pitilde} and \eqref{pi_s}.

\subsection{Proof of Proposition \ref{Proposition::0<p<1}} \label{App::0<p<1}
Before proving Proposition~\ref{Proposition::0<p<1}, we enunciate and demonstrate two lemmas that will be useful for the proof. The first one will help us in understand how to remove two crossing arcs without breaking the TSP cycle into multiple ones. The second one, instead will prove that removing a crossing between two arcs will always lower the total number of crossing in the TSP cycle.
\begin{lemma}
	\label{Lemma::0<p<1}
	Given an Hamiltonian cycle with its edges drawn as arcs in the upper half-plane, let us consider two of the arcs that cannot be drawn without crossing each other. Then, this crossing can be removed only in one way without splitting the original cycle into two disjoint cycles; moreover, this new configuration has a lower cost than the original one.   
\end{lemma}
\begin{proof}
	Let us consider a generic oriented Hamiltonian cycle and let us suppose it contains a matching as in figure:
	\begin{figure}[h!]
		\centering
		\includegraphics[width=0.4\columnwidth]{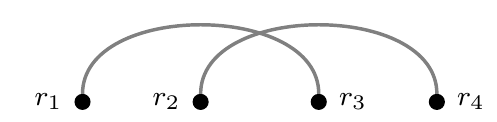}
	\end{figure}

	\noindent There are two possible orientations for the matching that correspond to this two oriented Hamiltonian cycles: 
	
	\begin{enumerate}
		\item $(C_1r_1r_3C_2r_2r_4C_3)\,,$
		\item $(C_1r_1r_3C_2r_4r_2C_3)\,,$
	\end{enumerate}
	where $C_1$, $C_2$ and $C_3$ are paths (possibly visiting other points of our set).
	The other possibilities are the dual of this two, and thus they are equivalent. In both cases, a priori, there are two choices to replace this crossing matching $(r_1, r_3)$, $(r_2,r_4)$ with a non-crossing one: $(r_1, r_2)$, $(r_3, r_4)$ or $(r_1, r_4)$, $(r_2, r_3)$. We now show, for the two possible prototypes of Hamiltonian cycles, which is the right choice for the non-crossing matching, giving a general rule. Let us consider case 1: here, if we replace the crossing matching with $(r_1, r_4)$, $(r_2, r_3)$, the cycle will split; in fact we would have two cycles:  $(C_1r_1r_4C_3)$ and $(r_3C_2r_2)$. Instead, if we use the other non-crossing matching, we would have: $(C_1r_1r_2[C_2]^{-1}r_3r_4C_3)$. This way we have removed the crossing without splitting the cycle. Let us consider now case 2: in this situation, using $(r_1, r_4)$, $(r_2, r_3)$ as the new matching, we would have: $(C_1r_1r_4[C_2]^{-1}r_3r_2C_3)$; the other matching, on the contrary, gives: $(C_1r_1r_2C_3)$ and $(r_3C_2r_4)$.
	
	The general rule is the following: given the oriented matching, consider the four oriented lines going inward and outward the node. Then, the right choice for the non-crossing matching is obtained joining the two couples of lines with opposite orientation.
	
	Since the difference between the cost of the original cycle and the new one simply consists in the difference between a crossing matching and a non-crossing one,  this is positive when $0<p<1$, as shown in~\cite{Caracciolo:159}.
\end{proof}

Now we deal with the second point: given an Hamiltonian cycle, in general it is not obvious that replacing non-crossing arcs with a crossing one, the total number of intersections increases. Indeed there could be the chance that one or more crossings are removed in the operation of substituting the matching we are interested in. 
Notice that two arcs forms a matching of 4 points. Therefore, from now on, we will use expressions like ``crossing matching'' (``non-crossing matching'') and ``two crossing arcs'' (``two non-crossing arcs'') indifferently.
We now show that it holds the following

\begin{lemma}
	\label{Lemma::intersections}
	Given an Hamiltonian cycle with a matching that is non-crossing, if it is replaced by a crossing one, the total number of intersections always increases. Vice versa, if a crossing matching is replaced by a non-crossing one, the total number of crossings always decreases.  
\end{lemma} 
\begin{proof}
	This is a topological property we will prove for cases. To best visualize crossings,
	we change the graphical way we use to represent the complete graph that underlies the problem: now the nodes are organized along a circle, in such a way that they are ordered clockwise (or, equivalently, anti-clockwise) according to the natural ordering given by the positions on the segment $[0,1]$. Links between points here are represented as straight lines. It is easy to see that a crossing as defined in Sec.~\ref{sec:p<0} corresponds to, in this picture, a crossing of lines.
	All the possibilities are displayed in Fig.~\ref{Fig::Crossings}, where we have represented with red lines the edges involved in the matching, while the other lines span all the possible topological configurations.
	
	\begin{figure*}[ht]
		\centering
		\includegraphics[width=1\columnwidth]{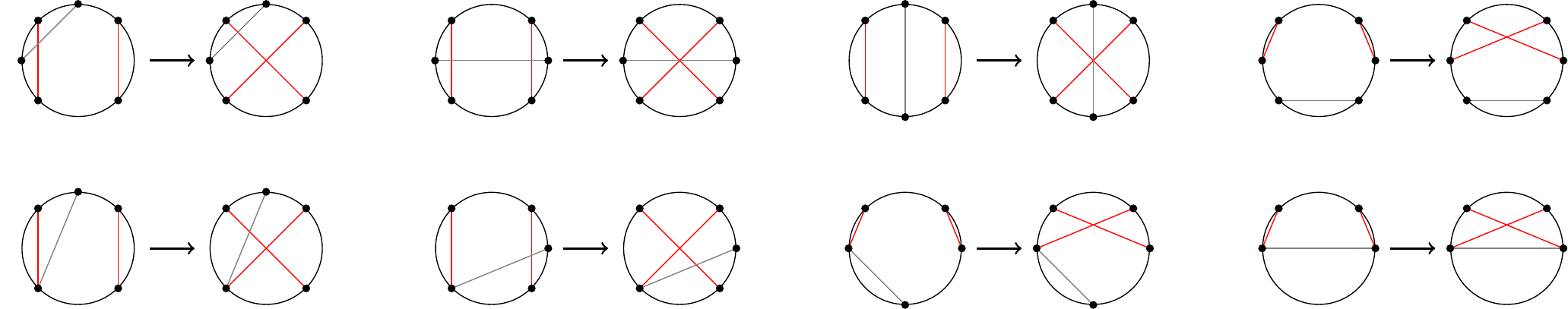}
		\caption{Replacing a non-crossing matching with a crossing one in an Hamiltonian cycle always increase the number of crossings. Here we list all the possible topological configurations one can have.}
		\label{Fig::Crossings}
	\end{figure*}
\end{proof}

Now we can prove Proposition~\ref{Proposition::0<p<1}:
\begin{proof}
	Consider a generic Hamiltonian cycle and draw the connections between the points in the upper half-plane. Suppose to have an Hamiltonian cycle where there are, let us say, $n$ intersections between edges. Thanks to Lemma~\ref{Lemma::0<p<1}, we can swap two crossing arcs with a non-crossing one without splitting the Hamiltonian cycle. As shown in Lemma~\ref{Lemma::intersections}, this operation lowers always the total number of crossings between the edges, and the cost of the new cycle is smaller than the cost of the starting one. Iterating this procedure, it follows that one can find a cycle with no crossings. 
	Now we prove that there are no other cycles out of $h^*$ and its dual with no crossings. This can be easily seen, since $h^*$ is the only cycle that visits all the points, starting from the first, in order. This means that all the other cycles do not visit the points in order and, thus, they have a crossing, due to the fact that the point that is not visited in a first time, must be visited next, creating a crossing.
\end{proof}

\subsection{Proof of Proposition \ref{prop<01}} \label{App::p<0 even N}
To complete the proof given in the main text, we need to discuss two points. Firstly, we address which is the correct move that swap a non-crossing matching with a crossing one; thanks to Lemma~\ref{Lemma::intersections}, by performing such a move one always increases the total number of crossings. Secondly we prove that there is only one Hamiltonian cycle to which this move cannot be applied (and so it is the optimal solution).

We start with the first point: consider an Hamiltonian cycle with a matching that is non-crossing, then the possible situations are the following two:
\begin{figure}[h!]
	\centering
	\includegraphics[width=0.4\columnwidth]{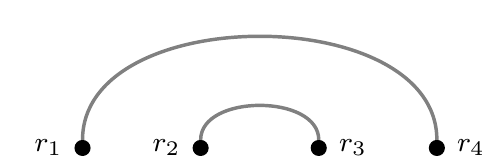} 
\end{figure}
\begin{figure}[h]
	\centering
	\includegraphics[width=0.4\columnwidth]{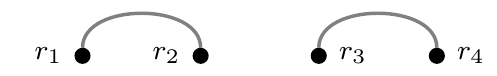}
\end{figure}

\noindent For the first case there are two possible independent orientations:

\begin{enumerate}
	\item $(r_1r_4C_2r_2r_3C_3)\,,$
	\item $(r_1r_4C_2r_3r_2C_3)\,.$
\end{enumerate}

If we try to cross the matchings in the first cycle, we obtain $(r_1r_3C_3)(r_2[C_2]^{-1}r_4)$, and this is not anymore an Hamiltonian cycle. On the other hand, in the second cycle, the non-crossing matching can be replaced by a crossing one without breaking the cycle: $(r_1r_3[C_2]^{-1}r_4r_2C_3)$. For the second case the possible orientations are:
\begin{enumerate}
	\item $(r_1r_2C_2r_4r_3C_3)\,,$
	\item $(r_1r_2C_2r_3r_4C_3)\,.$
\end{enumerate}
By means of the same procedure used in the first case, one finds that the non-crossing matching in the second cycle can be replaced by a crossing one without splitting the cycle, while in the first case the cycle is divided by this operation.\\

The last step is the proof that the Hamiltonian cycle given in Proposition \ref{prop<01} has the maximum number of crossings.

	Let us consider an Hamiltonian cycle $h[\sigma] = \left( r_{\sigma(1)}, \dots, r_{\sigma(N)} \right)$ on the complete graph $\mathcal{K}_N$. We now want to evaluate what is the maximum number of crossings an edge can have depending on the permutation $\sigma$. Consider the edge connecting two vertices $r_{\sigma(i)}$ and $r_{\sigma(i+1)}$: obviously both the edges $(r_{\sigma(i-1)}, r_{\sigma(i)})$ and $(r_{\sigma(i+1)}, r_{\sigma(i+2)})$ share a common vertex with ($r_{\sigma(i)}, r_{\sigma(i+1)}$), therefore they can never cross it. So, if we have $N$ vertices, each edge has $N-3$ other edges that can cross it. Let us denote with $\mathcal{N}\left[\sigma(i)\right]$ the number of edges that cross the edge  $(r_{\sigma(i)}, r_{\sigma(i+1)})$ and let us define the sets:
\be
A_j:=
\begin{cases}
	\{r_k\}_{k=\sigma(i)+1\Mod{N},\dots,\sigma(i+1)-1\Mod{N}}  & \hbox{for \, } j=1\\
	\{r_k\}_{k=\sigma(i+1)+1\Mod{N},\dots,\sigma(i)-1\Mod{N}}  & \hbox{for \, } j=2\\	 
\end{cases}
\label{Sets}
\ee
These two sets contain the points between $r_{\sigma(i)}$ and $r_{\sigma(i+1)}$. In particular, the maximum number of crossings an edge can have is given by:
\be	
\max(\mathcal{N}\left[\sigma(i)\right])=
\begin{cases}
	2\min_{j}|A_j|  & \hbox{for \, } |A_1|\not = |A_2|\\
	2|A_1| - 1  & \hbox{for \, } |A_1| = |A_2|\\	
\end{cases}\label{max}
\ee
This is easily seen, since the maximum number of crossings an edge can have is obtained when all the points belonging to the smaller between $A_1$ and $A_2$ contributes with two crossings. This cannot happen when the cardinality of $A_1$ and $A_2$ is the same because at least one of the edges departing from the nodes in $A_1$ for example, must be connected to one of the ends of the edge $(r_{\sigma(i)}, r_{\sigma(i+1)})$, in order to have an Hamiltonian cycle. Note that this case, i.e. $|A_1| = |A_2|$ can happen only if $N$ is even.

Consider the particular case such that $\sigma(i)=a$ and $\sigma(i+1)=a+\frac{N-1}{2}\pmod{N}$ or $\sigma(i+1)=a+\frac{N+1}{2}\pmod{N}$. Then \eqref{max} in this cases is exactly equal to $N-3$, which means that the edges $(r_a, r_{a+\frac{N-1}{2}\pmod{N}})$ and $(r_a, r_{a+\frac{N+1}{2}\pmod{N}})$ can have the maximum number of crossings if the right configuration is chosen.\\
Moreover, if there is a cycle such that every edge has $N-3$ crossings, such a cycle is unique, because the only way of obtaining it is connecting the vertex $r_a$ with $r_{a+\frac{N-1}{2}\pmod{N}}$ and $r_{a+\frac{N+1}{2}\pmod{N}}, \forall a$.\\

\section{The 2-factor and TSP solution for $p<0$ and even $N$} \label{B}

We start considering here the 2-factor solution for $p<0$ in the even-$N$ case. After that, we proof Proposition~\ref{TSP_p<0}. 

In the following we will say that, given a permutation $\sigma \in \mathcal{S}_{N}$, the edge $(r_{\sigma(i)}, r_{\sigma(i+1)})$ has length $L \in \mathbb{N}$ if: 
\be
L=\mathcal{L}(i):=\min_j|A_j(i)| 
\ee
where $A_j(i)$ was defined in equation (\ref{Sets}).

\subsection{$N$ is a multiple of 4}
Let us consider the sequence of points $\mathcal{R} = \{r_i\}_{i=1,\dots,N}$ of $N$ points, with $N$ a multiple of 4, in the interval $[0,1]$, with $r_1 \le \dots \le r_N$, consider the permutations $\sigma_j$, $j=1, 2$ defined by the following cyclic decomposition:
\begin{subequations}\label{N|4}
\begin{equation}
\begin{split}
\sigma_1 & = (r_1, r_{\frac{N}{2}+1}, r_2, r_{\frac{N}{2}+2}) \dots (r_a, r_{a+\frac{N}{2}}, r_{a+1}, r_{a+\frac{N}{2}+1}) \dots(r_{\frac{N}{2}-1}, r_{N-1}, r_{\frac{N}{2}}, r_{N}) \\
\end{split}
\end{equation}
\begin{equation}
\begin{split}
\sigma_2 & = (r_1, r_{\frac{N}{2}+1}, r_N, r_{\frac{N}{2}})\dots(r_a, r_{a+\frac{N}{2}}, r_{a-1}, r_{a+\frac{N}{2}-1}) \dots(r_{\frac{N}{2}-1}, r_{N-1}, r_{\frac{N}{2}-2}, r_{N-2})
\end{split}
\end{equation}
\end{subequations}
for integer $a = 1, \dots,  \frac{N}{2}-1$. Defined $h^*_1:=h[\sigma_1]$ and $h^*_2:=h[\sigma_2]$, it holds the following:

\begin{pros}\label{b1}
	$h^*_1$ and $h^*_2$ are the 2-factors that contain the maximum number of crossings between the arcs.
\end{pros}

\begin{proof}
	   
	An edge can be involved, at most, in $N-3$ crossing matchings. In the even N case, this number is achieved by the edges of the form $(r_a, r_{a+\frac{N}{2}\pmod{N}})$, i.e. by the edges of length $\frac{N}{2}-1$. There can be at most $\frac{N}{2}$ edges of this form in a 2-factor. Thus, in order to maximize the number of crossings, the other $\frac{N}{2}$ edges must be of the form $(r_a, r_{a+\frac{N}{2}+1\pmod{N}})$ or $(r_a, r_{a+\frac{N}{2}-1\pmod{N}})$, i.e. of length $\frac{N}{2}-2$. It is immediate to verify that both $h^*_1$ and $h^*_2$ have this property; we have to prove they are the only ones with this property.\\
	Consider, then, to have already fixed the $\frac{N}{2}$ edges $(r_a, r_{a+\frac{N}{2}\pmod{N}}),   \forall a\in [N]$. Suppose to have fixed also the edge $(r_1, r_{\frac{N}{2}})$ (the other chance is to fix the edge $(r_1, r_{\frac{N}{2}+2})$: this brings to the other 2-factor). Consider now the point $r_{\frac{N}{2}+1}$: suppose it is not connected to the point $r_N$, but to the point $r_2$, i.e., it has a different edge from the cycle $h^*_2$. We now show that this implies it is not possible to construct all the remaining edges of length $\frac{N}{2}-2$. Consider, indeed, of having fixed the edges $(r_1, r_{\frac{N}{2}})$ and $(r_2, r_{\frac{N}{2}+1})$ and focus on the vertex $r_{\frac{N}{2}+2}$: in order to have an edge of length $\frac{N}{2}-2$, this vertex must be connected either with $r_1$ or with $r_3$, but $r_1$ already has two edges, thus, necessarily, there must be the edge $(r_{\frac{N}{2}+2}, r_3)$. By the same reasoning, there must be the edges $(r_{\frac{N}{2}+3}, r_4)$, $(r_{\frac{N}{2}+4}, r_5), \dots, (r_{N-2}, r_{\frac{N}{2}-1})$. Proceeding this way, we have constructed $N-1$ edges; the remaining one is uniquely determined, and it is $(r_{N-1}, r_{N})$, which has null length.\\
	Therefore the edge $(r_2, r_{\frac{N}{2}+1})$ cannot be present in the optimal 2-factor and so, necessarily, there is the edge $(r_{\frac{N}{2}+1}, r_N)$; this creates the cycle $(r_1, r_{\frac{N}{2}}, r_N, r_{\frac{N}{2}+1})$. Proceeding the same way on the set of the remaining vertices $\{r_2, r_3,\dots,r_{\frac{N}{2}-1}, r_{\frac{N}{2}+2},\dots, r_{N-1}\}$, one finds that the only way of obtaining $\frac{N}{2}$ edges of length $\frac{N}{2}-1$ and $\frac{N}{2}$ edges of length $\frac{N}{2}-2$ is generating the loop coverings of the graph $h^*_1$ or $h^*_2$.
\end{proof}

Proposition \ref{b1}, together with the fact that the optimal 2-factor has the maximum number of crossing matchings, guarantees that the optimal 2-factor is either $h^*_1$ or $h^*_2$.

\begin{figure*}[ht]
	\begin{subfigure}[t]{0.49\linewidth}
		\centering
		\includegraphics[width=1\columnwidth]{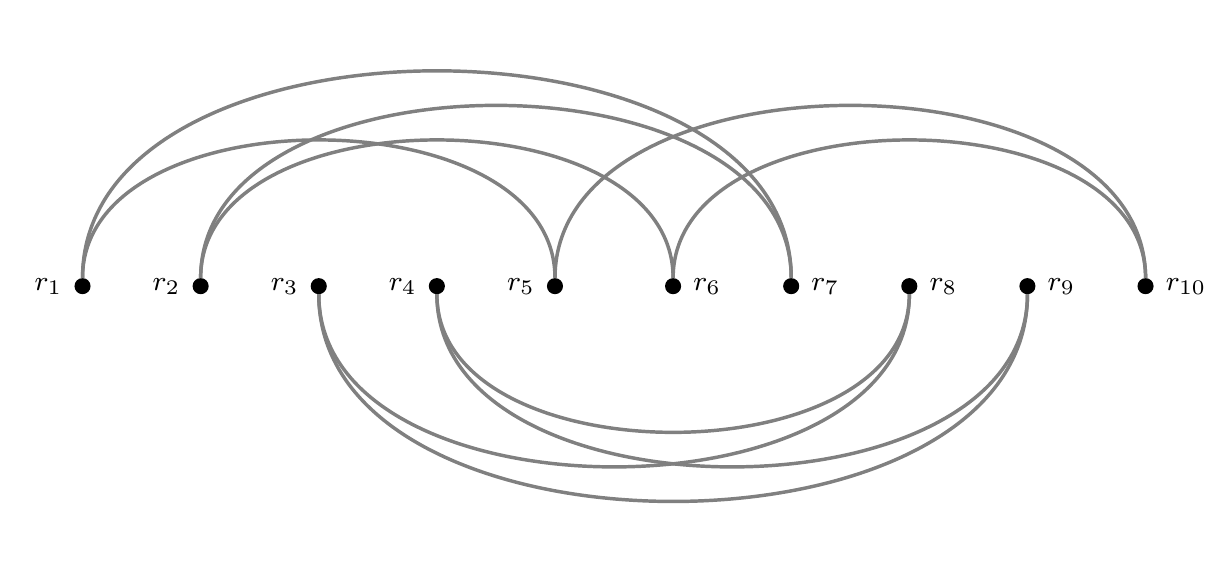}
		\caption{One of the optimal 2-factor solutions for $N=10$ and $p<0$; the others are obtainable cyclically permuting this configuration}
		\label{Fig::2factor_N10_line}
	\end{subfigure} \hfill
	\begin{subfigure}[t]{0.49\linewidth}
		\centering
		\includegraphics[width=0.5\columnwidth]{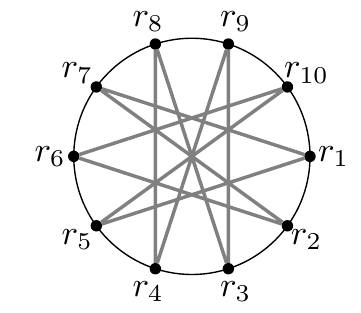}
		\caption{The same optimal 2-factor solution, but represented on a circle, where the symmetries of the solutions are more easily seen} \label{Fig::2factor_N10_circle}
	\end{subfigure}
	\caption{}
\end{figure*}

\subsection{$N$ is not a multiple of 4}
Let us consider the usual sequence $\mathcal{R} = \{r_i\}_{i=1,\dots,N}$ of $N$ points, with even $N$ but not a multiple of 4, in the interval $[0,1]$, with $r_1 \le \dots \le r_N$, consider the permutation $\pi$ defined by the following cyclic decomposition:
\be
\begin{split}\label{Nnot|4}
\pi & =(r_1, r_{\frac{N}{2}}, r_N, r_{\frac{N}{2}+1}, r_2, r_{\frac{N}{2}+2})(r_3, r_{\frac{N}{2}+3}, r_4, r_{\frac{N}{2}+4}) \dots(r_{\frac{N}{2}-2}, r_{N-1}, r_{\frac{N}{2}-1}, r_{N-2})
\end{split}
\ee
Defined 
\be
\pi_k(i):=\pi(i)+k \Mod N, \; k\in [0, N-1] \label{permut}
\ee
and
\be
h^*_k:=h[\pi_k]
\ee
the following proposition holds:
\begin{pros}\label{b2}
	$h^*_k$ are the 2-factors that contain the maximum number of crossings between the arcs.
\end{pros}
\begin{proof}
	Also in this case the observations done in the proof of Proposition~\ref{b1} holds. Thus, in order to maximize the number of crossing matchings, one considers, as in the previous case, the $\frac{N}{2}$ edges of length $\frac{N}{2}-1$, i.e. of the form $(r_a, r_{a+\frac{N}{2}\pmod{N}})$, and then tries to construct the remaining $\frac{N}{2}$ edges of length $\frac{N}{2}-2$, likewise the previous case. Again, if one fixes the edge $(r_1, r_{\frac{N}{2}})$, the edge $(r_2, r_{\frac{N}{2}+1})$ cannot be present, by the same reasoning done in the proof of Proposition B.1. The fact that, in this case, $N$ is not a multiple of 4 makes it impossible to have a 2-factor formed by 4-vertices loops, as in the previous case. The first consequence is that, given $\frac{N}{2}$ edges of length $\frac{N}{2}-1$, it is not possible to have $\frac{N}{2}$ edges of length $\frac{N}{2}-2$. In order to find the maximum-crossing solution, one has the following options:
	\begin{itemize}
		\item to take a 2-factor with $\frac{N}{2}$ edges of length $\frac{N}{2}-1$, $\frac{N}{2}-1$ edges of length $\frac{N}{2}-2$ and one edge of length $\frac{N}{2}-2$: in this case the theoretical maximum number of crossing matchings is $\frac{N(N-3)}{2}+(\frac{N}{2}-1)(N-4)+N-6=N^2-\frac{7N}{2}-2$;
		\item to take a 2-factor with $\frac{N}{2}-1$ edges of length $\frac{N}{2}-1$, $\frac{N}{2}+1$ edges of length $\frac{N}{2}-2$: in this case the theoretical maximum number of crossing matchings is $(\frac{N}{2}-1)(N-3)+(\frac{N}{2}+1)(N-4)=N^2-\frac{7N}{2}-1$.
	\end{itemize}
Clearly the second option is better, at least in principle, than the first one. The cycles $h^*_k$ belong to the second case and saturate the number of crossing matchings. Suppose, then, to be in this case. Let us fix the $\frac{N}{2}-1$ edges of length $\frac{N}{2}-1$; this operation leaves two vertices without any edge, and this vertices are of the form $r_a$, $r_{a+\frac{N}{2}\pmod{N}}, a\in[1,N]$ (this is the motivation for the degeneracy of solutions). By the reasoning done above, the edges that link this vertices must be of length $\frac{N}{2}-2$, and so they are uniquely determined. They form the 6-points loop $(r_a$, $r_{a-1+\frac{N}{2}\pmod{N}}$, $r_{N-1+a\pmod{N}}$, $r_{a+\frac{N}{2}\pmod{N}}$, $r_{a+1\pmod{N}}$, $r_{a+1+\frac{N}{2}\pmod{N}})$. The remaining $N-6$ points, since $4|(N-6)$, by the same reasoning done in the proof of Proposition~\ref{b1}, necessarily form the $\frac{N-6}{4}$ 4-points loops given by the permutations~\reff{permut}.
\end{proof}

Proposition \ref{b2}, together with the fact that the optimal 2-factor has the maximum number of crossing matchings, guarantees that the optimal 2-factor is such that $h^*\in \{h^*_k\}_{k=1}^N$. 

\subsection{Proof of Proposition~\ref{TSP_p<0}}
\begin{proof}
	Let us begin from the permutations that define the optimal solutions for the 2-factor, that is those given in Eqs.~\ref{N|4} if is $N$ a multiple of 4 and in Eq.~\ref{Nnot|4} otherwise.
	In both cases, the optimal solution is formed only by edges of length $\frac{N}{2}-1$ and of length $\frac{N}{2}-2$. Since the optimal 2-factor is not a TSP, in order to obtain an Hamiltonian cycle from the 2-factor solution, couples of crossing edges need to became non-crossing, where one of the two edges belongs to one loop of the covering and the other to another loop. Now we show that the optimal way of joining the loops is replacing two edges of length $\frac{N}{2}-1$ with other two of length $\frac{N}{2}-2$.
	Let us consider two adjacent 4-vertices loops, i.e. two loops of the form:
	\be
	(r_a, r_{a+\frac{N}{2}}, r_{a+1}, r_{a+\frac{N}{2}+1}), (r_{a+2}, r_{a+2+\frac{N}{2}}, r_{a+3}, r_{a+\frac{N}{2}+3})
	\ee
	and let us analyze the possible cases:
	\begin{enumerate}
		\item to remove two edges of length $\frac{N}{2}-2$, that can be replaced in two ways:
		\begin{itemize}
			\item either with an edge of length $\frac{N}{2}-2$ and one of length $\frac{N}{2}-4$; in this case the maximum number of crossings decreases by 4;
			\item or with two edges of length $\frac{N}{2}-3$; also in this situation the maximum number of crossings decreases by 4.
		\end{itemize}
		\item to remove one edge of length $\frac{N}{2}-2$ and one of length $\frac{N}{2}-1$, and also this operation can be done in two ways:
		\begin{itemize}
			\item either with an edge of length $\frac{N}{2}-2$ and one of length $\frac{N}{2}-3$; in this case the maximum number of crossings decreases by 3;
			\item or with an edge of length $\frac{N}{2}-3$ and one of length $\frac{N}{2}-4$; in this situation the maximum number of crossings decreases by 7.
		\end{itemize}
		\item the last chance is to remove two edges of length $\frac{N}{2}-1$, and also this can be done in two ways:
		\begin{itemize}
			\item either with two edges of length $\frac{N}{2}-3$; here the maximum number of crossings decreases by 6;
			\item or with two edges of length $\frac{N}{2}-2$; in this situation the maximum number of crossings decreases by 2. This happens when we substitute two adjacent edges of length $\frac{N}{2}-1$, that is, edges of the form $(r_a,r_{\frac{N}{2}+a\pmod{N}})$ and $(r_{a+1},r_{\frac{N}{2}+a+1\pmod{N}})$, with the non-crossing edges $(r_a,r_{\frac{N}{2}+a+1\pmod{N}})$ and $(r_{a+1},r_{\frac{N}{2}+a\pmod{N}})$
		\end{itemize}
	\end{enumerate}
	The last possibility is the optimal one, since our purpose is to find the TSP with the maximum number of crossings, in order to conclude it has the lower cost. Notice that the cases discussed above holds also for the 6-vertices loop and an adjacent 4-vertices loop when N is not a multiple of 4. We have considered here adjacent loops because, if they were not adjacent, then the difference in maximum crossings would have been even bigger.\\
	Now we have a constructive pattern for building the optimal TSP. Let us call $\mathcal{O}$ the operation described in the second point of (3). Then, starting from the optimal 2-factor solution, if it is formed by $n$ points, $\mathcal{O}$ has to be applied $\frac{N}{4}-1$ times if N is a multiple of 4 and $\frac{N-6}{4}$ times otherwise. In both cases it is easily seen that $\mathcal{O}$ always leaves two adjacent edges of length $\frac{N}{2}-1$ invariant, while all the others have length $\frac{N}{2}-2$. The multiplicity of solutions is given by the $\frac{N}{2}$ ways one can choose the two adjacent edges of length $\frac{N}{2}-1$. In particular, the Hamiltonian cycles $h^*_k$ saturates the maximum number of crossings that can be done, i.e., every time that $\mathcal{O}$ is applied, exactly 2 crossings are lost.\\
	We have proved, then, that $h^*_k$ are the Hamiltonian cycles with the maximum number of crossings. Now we prove that any other Hamiltonian cycle has a lower number of crossings. Indeed any other Hamiltonian cycle must have
	\begin{itemize}
		\item either every edge of length $\frac{N}{2}-2$;
		\item or at least one edge of length less than or equal to $\frac{N}{2}-3$.
	\end{itemize} 
	This is easily seen, since it is not possible to build an Hamiltonian cycle with more than two edges or only one edge of length $\frac{N}{2}-1$ and all the others of length $\frac{N}{2}-2$. It is also impossible to build an Hamiltonian cycle with two non-adjacent edges of length $\frac{N}{2}-1$ and all the others of length $\frac{N}{2}-2$: the proof is immediate.
	Consider then the two cases presented above: in the first case the cycle (let us call it $H$) is clearly not optimal, since it differs from $h^*_k, \forall k$ by a matching that is crossing in $h^*_k$ and non-crossing in $H$. Let us consider, then, the second case and suppose the shortest edge, let us call it $b$, has length $\frac{N}{2}-3$: the following reasoning equally holds if the considered edge is shorter. The shortest edge creates two subsets of vertices: in fact, called $x$ and $y$ the vertices of the edge considered and supposing $x<y$, there are the subsets defined by:
	\be
	A=\{r\in \mathcal{V}: x<r<y\}
	\ee
	\be
	B=\{r\in \mathcal{V}: r<x \vee r>y\}
	\ee
	Suppose, for simplicity, that $|A|<|B|$: then, necessarily $|A|=\frac{N}{2}-3$ and $|B|=\frac{N}{2}+1$. As an immediate consequence, there is a vertex in $B$ whose edges have both vertices in $|B|$. As a consequence, fixed an orientation on the cycle, one of this two edges and $b$ are obviously non-crossing and, moreover, have the right relative orientation so that they can be replaced by two crossing edges without splitting the Hamiltonian cycle. Therefore also in this case the Hamiltonian cycle considered is not optimal.
\end{proof}

\section{General distribution of points} \label{app::C}

In this section we shall consider a more general distribution of points.

Let choose the points in the interval $[0,1]$ according to the distribution $\rho$, which has no zero in the interval, and let
\be
\Phi_\rho(x) = \int_0^x dt\, \rho(t)
\ee
its {\em cumulative}, which is an increasing function with $\Phi_\rho(0) =0 $ and $ \Phi_\rho(1) = 1$.

In this case, the probability of finding the $l$-th point in the interval $[x, x+ dx]$ and the $s$-th point in the interval $[y, y+ dy]$ is given, for $s>l$  by
\begin{equation}
\begin{split}
p_{l,s}(x, y) \,d\Phi_\rho(x)\, d\Phi_\rho(y) = & \, \frac{\Gamma(N+1)}{\Gamma(l) \, \Gamma(s-l) \, \Gamma(N-s+1)} \\
& \, \Phi_\rho^{l-1}(x)  \left[\Phi_\rho(y)-\Phi_\rho(x)\right]^{s-l-1}\left[1-\Phi_\rho(y)\right]^{N-s} \\
& \, \theta(y-x) \,d\Phi_\rho(x)\, d\Phi_\rho(y)
\end{split}
\end{equation}
We have that, in the case $p>1$ 
\be
\begin{split}
	\overline{E_N[h^*]} = & \int \! d\Phi_\rho(x)\, d\Phi_\rho(y) \, (y-x)^p  \left[ p_{1,2}(x, y) + p_{N-1,N}(x, y) + \sum_{l=1}^{N-2} p_{l,l+2}(x, y)\right]\,
\end{split}
\ee
and
\begin{equation}
\begin{split}
\sum_{l=1}^{N-2} p_{l,l+2}(x, y) 
= \; & \frac{\Gamma(N+1)}{\Gamma(N-2)} \, \left[1-\Phi_\rho(y)+\Phi_\rho(x)\right]^{N-3} \left[\Phi_\rho(y)-\Phi_\rho(x)\right] \, \theta(y-x) 
\end{split}
\end{equation}
while
\begin{equation}
\begin{split}
&\left[p_{1,2}(x, y) + p_{N-1,N}(x, y)\right] 
= \frac{\Gamma(N+1)}{\Gamma(N-1)} \, \left[\left(1-\Phi_\rho(y)\right)^{N-2} + \Phi_\rho^{N-2}(x) \right] \, \theta(y-x) 
\end{split}
\end{equation}
For large $N$ we can make the approximation
\begin{equation}
\begin{split}
 \overline{E_N[h^*]} & \approx 
N^3\, \int\! d\Phi_\rho(x)\, d\Phi_\rho(y) \, (y-x)^p  \\
& \times \left[1-\Phi_\rho(y)+\Phi_\rho(x)\right]^{N} \left[\Phi_\rho(y)-\Phi_\rho(x)\right] \, \theta(y-x)
\end{split}
\end{equation}
and we remark that the maximum of the contribution to the integral comes from the region 
where $\Phi_\rho(y) \approx \Phi_\rho(x)$ and we make the change of variables
\be
\Phi_\rho(y) =  \Phi_\rho(x) + \frac{\epsilon}{N}
\ee
so that
\be
y = \Phi_\rho^{-1}\left[\Phi_\rho(x) + \frac{\epsilon}{N} \right] \approx x + \frac{\epsilon}{N \rho(x)}
\ee
and we get
\begin{align}
\overline{E_N[h^*]}  \approx  &\, 
N^3\, \int\, d\Phi_\rho(x) \int_0^\infty \frac{d\epsilon}{N} \left[\frac{\epsilon}{N \rho(x)}\right]^p \frac{\epsilon}{N} \, e^{- \epsilon} = \frac{\Gamma(p+2)}{N^{p-1}} \int dx\, \rho^{1-p}(x) \, .
\end{align}
When $0<p<1$ 
\be
\begin{split}
	\overline{E_N[h^*]} = & \int \! d\Phi_\rho(x)\, d\Phi_\rho(y) \, (y-x)^p \left[ p_{1,N}(x, y) +  \sum_{l=1}^{N-1} p_{l,l+1}(x, y)\right]
\end{split}
\ee
and
\begin{equation}
\begin{split}
&\sum_{l=1}^{N-1} p_{l,l+1}(x, y)  
=  N (N-1)  \, \left[1-\Phi_\rho(y)+\Phi_\rho(x)\right]^{N-2} \, \theta(y-x) 
\end{split}
\end{equation}
while
\begin{equation}
\begin{split}
p_{1,N}(x, y) 
= N (N-1) \, \left[\Phi_\rho(y)- \Phi_\rho(x) \right]^{N-2}\, \theta(y-x) 
\end{split}
\end{equation}
For large $N$ we can make the approximation
\begin{equation}
\begin{split}
\overline{E_N[h^*]}  \approx  &\, N^2 \int \! d\Phi_\rho(x)\, d\Phi_\rho(y) \, (y-x)^p \left[1-\Phi_\rho(y)+\Phi_\rho(x)\right]^{N} \, \theta(y-x) \\
\approx  &\, N^2\, \int\, d\Phi_\rho(x) \int_0^\infty \frac{d\epsilon}{N} \left[\frac{\epsilon}{N \rho(x)}\right]^p \, e^{- \epsilon} =  \frac{\Gamma(p+1)}{N^{p-1}} \int dx\, \rho^{1-p}(x) \, .
\end{split}
\end{equation}
Indeed the other term, for large $N$, gives a contribution 
\be
N^2\, \int (y-x)^p \left[\Phi_\rho(y)-\Phi_\rho(x)\right]^{N} \, \theta(y-x) \,d\Phi_\rho(x)\, d\Phi_\rho(y)
\ee
so that, we will set
\be
\Phi_\rho(y) = 1 - \frac{\epsilon}{N} \, , \qquad \Phi_\rho(x) = \frac{\delta}{N}\, , \qquad y-x \approx 1
\ee
and therefore we get a contribution
\be
\int_0^\infty d\epsilon \, e^{-\epsilon} \, \int_0^\infty d\delta \, e^{-\delta} =1\, 
\ee
which is of the same order of the other term only at $p=1$.

\section{Calculation of the second moment of the optimal cost distribution} \label{app::2Moment}
In this Appendix we compute the second moment of the optimal cost distribution. We will restrict for simplicity to the $p>1$ case, where
\be
E_N[h^*]=|r_{2}-r_1|^p+|r_{N}-r_{N-1}|^p+\sum_{i=1}^{N-2}|r_{i+2}-r_{i}|^p \,.
\label{OptimalCost}
\ee
We begin by writing the probability distribution for $N$ ordered points
\be
\rho_N(r_1,\dots,r_N)=N! \prod_{i=0}^{N}\theta(r_{i+1}-r_i)
\ee
where we have defined $r_0\equiv0$ and $r_{N+1}\equiv 1$. The joint probability distribution of their spacings
\begin{equation}
\varphi_i \equiv r_{i+1} - r_i \,,
\end{equation}
is, therefore
\begin{equation}
\rho_N(\varphi_0, \dots , \varphi_{N}) = N! \, \delta\left[ \sum_{i=0}^{N} \varphi_i = 1\right]\, \prod_{i=0}^{N} \theta(\varphi_i)\, .
\end{equation}
If $\{i_1, i_2, \dots, i_k\}$ is a generic subset of $k$ different indices in $\{0,1,\dots,N\}$, we soon get the marginal distributions
\begin{equation}
\label{DistributionSpacings}
\rho_N^{(k)} (\varphi_{i_1}, \dots , \varphi_{i_k}) = \frac{N!}{(N-k)!} \left( 1 - \sum_{n=1}^{k}\varphi_{i_n} \right)^{N-k} \theta \left( 1 - \sum_{n=1}^{k}\varphi_{i_n} \right) \, \prod_{n=1}^{k}\theta(\varphi_{i_n}) \,.
\end{equation}
Developing the square of \eqref{OptimalCost} one obtains $N^2$ terms, each one describing a particular configuration of two arcs connecting some points on the line. We will denote by $\chi_1$ and $\chi_2$ the length of these arcs; they can only be expressed as a sum of 2 spacings or simply as one spacing. Because the distribution~\eqref{DistributionSpacings} is independent of $i_1$, $\dots$, $i_k$, these terms can be grouped together on the base of their topology on the line with a given multiplicity. All these terms have a weight that can be written as
\begin{equation}
\label{GeneralTerm}
\int_{0}^{1} d \chi_1 \, d \chi_2 \; \chi_1^p \, \chi_2^p \, \rho(\chi_1,\chi_2)
\end{equation}
where $\rho$ is a joint distribution of $\chi_1$ and $\chi_2$. Depending on the term in the square of~\eqref{OptimalCost} one is taking into account, the distribution $\rho$ takes different forms, but it can always be expressed as in function of the distribution~\eqref{DistributionSpacings}. As an example, we show how to calculate $\overline{|r_{3}-r_1|^{p}|r_{4}-r_2|^{p}}$. In this case $\rho(\chi_1, \chi_2)$ takes the form
\begin{multline}
\rho(\chi_1, \chi_2) = \int d\varphi_1 \, d\varphi_2 \, d\varphi_3 \, \rho^{(3)}_N(\varphi_1, \varphi_2, \varphi_3) \delta\left( \chi_1 - \varphi_1 - \varphi_2 \right) \delta\left( \chi_2 - \varphi_2 - \varphi_3 \right) \\
= N(N-1) \left[ (1-\chi_1)^{N-2}\theta(\chi_1) \theta(\chi_2- \chi_1) \theta(1-\chi_2) \right. \\
+ (1-\chi_2)^{N-2}\theta(\chi_2) \theta(\chi_1- \chi_2) \theta(1-\chi_1) \\
\left. - (1-\chi_1-\chi_2)^{N-2} \theta(\chi_1) \theta(\chi_2) \theta(1-\chi_1-\chi_2)  \right]\,,
\end{multline}
that, plugged into~\eqref{GeneralTerm} gives
\begin{equation}
\overline{|r_{3}-r_1|^{p}|r_{4}-r_2|^{p}} = \frac{\Gamma(N+1) \left[\Gamma (2p+3)- \Gamma (p+2)^2\right]}{(p+1)^2 \Gamma (N+2 p+1)} \,.
\end{equation}
All the other terms contained can be calculated the same way; in particular there are 7 different topological configurations that contribute. After having counted how many times each configuration appears in $(E_N[h^*])^2$, the final expression that one gets is
\begin{multline}
\label{SecondMoment}
\overline{(E_N[h^*])^2}=\frac{\Gamma(N+1)}{\Gamma (N+2 p+1)} \Biggl[ 4(N-3) \Gamma (p+2) \Gamma (p+1) \Biggr. \\
 + \left((N-4) (N-3) (p+1)^2-2 N+8\right) \Gamma(p+1)^2 +\\
\left.+\frac{[N (2 p+1)(p+5)-4 p (p+5)-8]\, \Gamma (2 p+1)}{(p+1)} \right] \,.
\end{multline}

\bibliographystyle{unsrt}
\bibliography{AssignmentANDTsp}

\end{document}